\pdfoutput=1

\documentclass[reqno]{amsart}
\usepackage[foot]{amsaddr}
\usepackage{amssymb}
\usepackage[hyperindex,breaklinks]{hyperref}
\usepackage{eucal}
\usepackage[T2A,OT1]{fontenc}

\usepackage{csquotes}
\usepackage[russian,french,english]{babel}

\usepackage{bm}
\usepackage{mathtools}
\usepackage[sorting=none,
autolang=other]{biblatex}
\DeclareFieldFormat{eprint:CiteSeerX}{%
  CiteSeerX: \href{https://citeseerx.ist.psu.edu/viewdoc/summary?doi=#1}{#1}%
}

\usepackage{enumitem}
\setlist[itemize]{itemindent=2em, leftmargin=0pt}
\setlist[enumerate]{label = \textup{(\roman*)}}
\setlist[description]{itemindent=2em, leftmargin=0pt}

\AtBeginBibliography{\small}

\theoremstyle{definition}\newtheorem{Def}{Definition}[section]
\theoremstyle{definition}\newtheorem{Asmp}{Assumption}[section]
\theoremstyle{definition}\newtheorem{Ex}{Example}[section]
\theoremstyle{remark}
\theoremstyle{remark}
\theoremstyle{remark}
\theoremstyle{definition}
\theoremstyle{plain}\newtheorem{Th}{Theorem}[section]
\theoremstyle{plain}\newtheorem{Fact}{Fact}[section]
\theoremstyle{plain}\newtheorem{Le}{Proposition}[section]

\newcommand{\df}{\buildrel{}_\mathrm{def}\over=}

\DeclareMathOperator{\marg}{marg}

\author[Alexey V. Osipov]{Alexey V. Osipov$^{\ast,\ddagger}$}   

\author[Nikolay N. Osipov]{Nikolay N. Osipov$^{\dagger,\S,\P}$}     

\address{$^\ast$Swiss Re Institute, Z\"urich, Switzerland}
\address{$^\dagger$St. Petersburg State University, 
St. Petersburg, Russia}
\address{$^\ddagger$LLC Usetech Professional, Moscow, Russia}
\address{$^\S$St. Petersburg Department of V.~A.~Steklov Institute of Mathematics of the Russian Academy of Sciences, St. Petersburg, Russia}
\address{$^\P$HSE University, International Laboratory of Game Theory and Decision Making, St. Petersburg, Russia}
\email{nicknick AT pdmi DOT ras DOT ru}

\thanks{
The work was supported by the Russian Science Foundation grant 24-11-00087}

\hypersetup{
pdfstartview = FitH,colorlinks = true, linkcolor = blue, citecolor = red}

\title[Collective intelligence in science]{
Collective intelligence in science:\\ direct elicitation of diverse 
information 
from 
experts with unknown information structure}

\begin{document}
\begin{abstract}
Suppose we need a deep collective analysis of an open scientific problem: there is a complex scientific hypothesis and a
large online group of mutually unrelated experts with relevant private information of a diverse and unpredictable nature. This information may be results of
experts' individual experiments, 
original reasoning
of some of them, results of AI systems they use, etc. We propose a simple mechanism based on 
a self-resolving play-money prediction market entangled with a chat. 
We show that such a system can
easily be brought to an equilibrium where participants directly share their private information on the hypothesis through the chat and trade as if the market were resolved in accordance
with the truth of the hypothesis. This approach will lead to efficient aggregation of relevant
information in a completely interpretable form even if the ground truth cannot be established and experts initially know nothing about each other and cannot perform complex Bayesian calculations. Finally, by rewarding the experts with some real assets proportionally to the play money they end up with,
we can get an innovative way to fund large-scale collaborative studies of any type.

\end{abstract}

\maketitle

The instruments like the Bayesian truth serum~\cite{Pr2004}, its refinements~\cite{WiPa2012,RaFa2013,RaFa2014}, Bayesian markets~\cite{Ba2017}, or the peer prediction method~\cite{MiReZe2005}
can effectively incentivize experts to truthfully report their private signals 
in settings where their reports respond to a single question and cannot be directly verified. The reported signals can be binary in response to questions like ``Are you satisfied with your life?'', multi-valued, or even continuous in response to questions like ``What is your probability estimate that humanity will survive past the year 2100?''.
Such universal mechanisms are effective and robust in most cases. However, there is a particular but important case where we need a dedicated, more specialized tool. 



Namely, suppose we need a deep collective analysis of an open scientific problem. To be more specific, suppose there is a complex scientific hypothesis~$H$ and a large online group of experts. In addition, momentarily assume that their pooled private information suffices to prove the validity of~$H$ (this assumption is by no means necessary in general). Such information may be of an arbitrarily complex nature: results of complex experiments, content of academic publications, sophisticated reasoning of some experts or LLMs they use, with the LLMs' answers depending heavily on the skills of the questioners, etc. Namely, not only the information itself, but the information structures of experts may be diverse and unpredictable, and the experts may initially know nothing about how each other's information is arranged. 
In this setting, if we attempt, for example, to iteratively utilize the Bayesian truth serum, repeatedly eliciting the experts' probabilities of~$H$ and publicizing them in order to reach a common consensus probability~\cite{Au1976,GePo1982}, then the result will hardly be equal to~$1$. Below we discuss this issue from a theoretical perspective: in terms of~\cite{GePo1982}, the direct and indirect communication equilibria do not have to coincide 
in our setting. In addition to the fact that the consensus probability may be far from the probability conditional on the pooled information, it is merely a number that cannot be interpreted by an analyst who is ignorant of the experts' information structure. Thus, the result may be almost useless scientifically. 

The same problems would arise if we tried, relying on the theory 
\cite{MKePa1986, NiBrGe1990, FeFoPe2005, ChMuCh2006}, to get a consensus probability by means of a conventional prediction market with betting on the truth of~$H$. Again, since the situation in question can hardly be described through finite and commonly known information structure, we can expect that the result would tend to be far from the direct communication 
equilibrium and, in addition, would be uninterpretable.
Moreover, two new issues would arise. First, we need to know the ground truth in order to resolve a conventional prediction market. This looks like a vicious circle: in order to reward experts for their attempt to validate~$H$, we need a conclusive proof or disproof of~$H$. Second, various empirical studies like \cite{MKePa1990,AlKiPf2009} demonstrate that even if 
the information structure is commonly known and information pooling is theoretically achievable, this may not be the case in reality because experts need to be perfect Bayesians in order to effectively elicit information from aggregates such as prediction market prices.

Nevertheless, we claim the following. Consider a self-resolving play-money prediction market where each participant
starts with one and the same amount of play money and
winners are determined by a binary random generator with probability parameter equal to the market final price. 
We argue that such a market can easily be brought to an equilibrium where participants directly share their private information on~$H$ via chat and trade as if the market were resolved in accordance with the truth of~$H$. 
This approach will lead to efficient aggregation of diverse information in a completely interpretable form even if the information structure is not initially commonly known. Furthermore, the analyst do not need to know the ground truth in order to reward experts, and experts do not need to perform Bayesian calculations in order to update their information. For example, if $H$ is a true hypothesis and the experts' pooled information suffices to prove it, then it will be proved even in situations where we cannot hope to achieve this by other methods and, what is more, the proof will be provided explicitly. 
By rewarding the winners with some real assets proportionally to the play money they end up with, we can effectively incentivize large-scale collaborative research of any type. 

We borrow the term ``self-resolving prediction market'' from \cite{AhWi2018}, where similar markets are experimentally studied in another context. Another example of self-resolving markets is described in \cite{DaKiLo2011}. 
The idea of a mechanism that provides incentives to elicit explanations for beliefs from experts  
was first proposed in our working paper~\cite{OsOs2023}. We thereby anticipated the later trend \cite{RoWiSch2023, SrKaBa2025}. In the present article, we will describe in detail and formalize the aforementioned idea in the standard setting of epistemic game theory. Before that, we will demonstrate that when it comes to a complex scientific hypothesis, the presence of direct information exchange is not just a useful feature, but an “all or nothing” question. It is also worth noting that in the aforementioned working paper \cite[Appendix C]{OsOs2023}, we additionally considered an extremely general setting in which experts were assumed to be neither Bayesian with respect to their information, nor risk neutral, nor rational in the sense of the von Neumann--Morgenstern axioms, and the exchange of information was 
an arbitrary flow of signals that were not specified a priori. Not fully formalized reasoning based on the analysis of extensive data from a conventional prediction market showed that even in such a setting, the desired behavior of experts can be expected.

\section{Problem}\label{sec:problem} 

Throughout this section, we consider a simple, but illustrative example of a situation where the existing instruments and the current consensus theory are insufficient. We emphasize that this is only an introductory example and our subsequent results address the general situation where experts' knowledge may be 
arbitrary.

\begin{enumerate}[label = (exm)]
\item\label{it:ex}
Suppose there are experts with a common prior~$\pi$ over the set~$\Omega$ of possible states of the world. Suppose $H$, $A$, and $B$ are pairwise disjoint 
nonnull events in $\Omega$, and~$H$ is a true hypothesis (i.e., it contains the true state of the world).
Proceeding as in the most basic examples, we suppose one of the experts knows that~$A$ is false, another expert knows that $B$ is false, and all the others possess only the public information, which 
is incorporated in the common prior. 
\end{enumerate}
In addition to the experts, consider an analyst who has no private information and know nothing about the experts' information structure, but would like to effectively assess the plausibility of~$H$. 
Suppose the analyst has some way to incentivize the experts to give true answers about their beliefs about~$H$, for example, by means of 
the Bayesian truth serum \cite{Pr2004,WiPa2012,RaFa2013,RaFa2014}. 
In this case, 
the answers of two informed experts will be $\pi(H\mid\neg A)$ and $\pi(H\mid\neg B)$, while all the other answers will be 
$\pi(H)$. Here $\neg A$ and~$\neg B$ are the complementary events of $A$ and $B$, respectively.
If, for example, 
\begin{equation}\label{eq:HAB1}
\pi(H) = \pi(A) = \pi(B) = \tfrac{1}{3},
\end{equation}
then the answers of both informed experts will be~$\frac{1}{2}$. Since the analyst knows nothing about the experts' information structure, 
he or she cannot interpret the answers. Instead, he or she can facilitate indirect communication between experts 
in order for them to interpret each other's beliefs. Namely, the analyst can make the answers public, repeat 
the 
survey to get the updated beliefs, make them public, repeat the survey, and so on. 
We show that the result of such a procedure will strongly depend on what the experts believe about each other's information structures. 

The information structures can be described~\cite{Au1976} as partitions of~$\Omega$: each expert is informed of that element 
of his or her partition that contains 
the true state of the world. 
In particular, $\neg A$ and $\neg B$ are elements of the partitions of the corresponding knowledgeable experts in the example under discussion. 
In this example, we also assume that $\pi$ is a non-atomic measure and do \emph{not} assume that the experts' information partitions are commonly known among them (as was the case in \cite{Au1976,GePo1982}). 
Such a setting is 
natural if we interpret $H$ as a complex scientific hypothesis, and the experts as a large online group of researchers who do not know each other and may possess information of arbitrary nature: 
results of experiments and measurements of arbitrary complexity and accuracy, 
answers of various LLMs to differently worded questions, 
content of academic publications, etc. 

We note that the non-atomicity of~$\pi$ by no means implies that experts simultaneously reflect on infinite number of separate heterogeneous objects. Although our example is simplified and does not map directly to the real world, this continuity assumption comes precisely from reality: for instance, the numerical results of a particular complex 
experiment rule out a continuum 
of its other results that were thought possible (cf. $A$), and leave a continuum 
of possibilities for the results of the other experiments (cf. $\neg A$). Informally, the non-atomicity of~$\pi$ means that it is impossible to guess which experiment was conducted based on the posterior probability of~$H$ alone (i.e., to guess an expert's information partition).

Next, situations with incomplete knowledge of each other's partitions can be formalized by means of types and hierarchies of beliefs \cite{Ha1967, MeZa1985} or, what is the same~\cite{BrDe1993}, 
through
new ``partitions'' ($\sigma$-algebras in fact) over a certain extension of~$\Omega$, which may be treated as commonly known structures. It is important to note that these new commonly known ``partitions'' may be infinite even if the initial ones were finite.

For simplicity, we make the following assumption on the experts' belief about each other's information structures.
\begin{enumerate}[label = (blf0)]
\item\label{it:sim0} It is commonly believed among the experts that the private information of 
any of them 
is 
one-directional, 
i.e., that each of them can, at most, determine the truth or falsity of 
some subset of $\neg H$. 
\end{enumerate}
We will assume~\ref{it:sim0} only in the context of~\ref{it:ex}. The framework~\ref{it:sim0} approximates a scenario where a hypothesis cannot be confirmed directly, but can be tested through various experiments. Each trial either increases the plausibility of the hypothesis if successful, or conclusively disproves it if not. For example, in drug development, various preclinical trials (in vitro and in vivo studies) serve this role: positive results 
increase confidence in the potential efficacy of a new
drug in future human clinical trials, while any negative result 
may preclude its viability and lead to cancellation of further trials at all.


Thus, it is common knowledge that each expert has either proved the hypothesis~$\neg H$, 
or ruled out some variants of its realization (maybe a null set of variants). 
Therefore, after the first survey 
and publicizing the results, two of which turn out to be greater than the prior $\pi(H)$ and the others coincide with~$\pi(H)$, it will become commonly known among the experts that some two nonnull events~$A$ and~$B$ within $\neg H$ are not true. Their priors $\pi(A)$ and $\pi(B)$ will also be calculated by the experts and will become commonly known among them (but not $A$ and $B$ themselves). The updated beliefs about~$H$ and, as a consequence, the results of the second survey 
will strongly depend on the experts' beliefs about 
overlapping of~$A$ and~$B$. One of the two extreme cases is where after the first iteration it will be common knowledge among experts that $A$ and $B$ represent completely different information and are therefore pairwise disjoint. In this case the second iteration will produce 
the final consensus 
coinciding with
$$
\pi(H\mid \neg A \cap \neg B) = \frac{\pi(H)}{1 - \pi(A) - \pi(B)}.
$$ 
Namely, in terms of~\cite{GePo1982} the revision process will take $2$ steps to converge and there will be no discrepancy between the direct and indirect communication equilibria.
In particular, if~\eqref{eq:HAB1} holds, then the analyst will see 
that $H$ is true with probability~$1$.

But from a Bayesian perspective, it is natural for the experts to expect some similarity in their private information,
and
the situation worsens once they 
doubt that their information is mutually complementary. 
The corresponding 
extreme case is where the experts have the false common belief that their private information has a single source and is not complementary at all, i.e., where after the first iteration it will be common belief that $A$ and $B$ are nested. Since~$\pi$ is non-atomic, this belief is always internally consistent no matter what $\pi(A)$ and $\pi(B)$ will be revealed to be. In this case, 
the second 
survey 
will produce the final consensus equal to
$$
\max\bigl(\pi(H\mid \neg A),\,\pi(H\mid \neg B)\bigr) = \frac{\pi(H)}{1 - \max\bigl(\pi(A),\pi(B)\bigr)}.
$$
In particular, if~\eqref{eq:HAB1} holds, then the consensus probability of $H$ will be $\tfrac{1}{2}$: the result will change from $1$ to $\tfrac{1}{2}$, i.e., from complete certainty to the maximum entropy! However, both extreme cases do not seem realistic, and in the case of~\eqref{eq:HAB1}, some intermediate value will be obtained if after the first iteration it will be common knowledge among the experts that each of them considers the possible values 
of~$\pi(A\cup B)$ to be uniformly distributed between the extreme values
$\tfrac{1}{3}$ and $\tfrac{2}{3}$. In this case, the second iteration will lead to the final consensus equal to the value
$$
\int_{\tfrac{1}{3}}^{\tfrac{2}{3}} \frac{dx}{1-x} = \log 2 \approx 0.69,
$$
which is between $\tfrac{1}{2}$ and $1$ and is also far from the pooling probability~$1$.

The underlying cause of such inefficiency is that 
the situation cannot be reduced to commonly known \emph{finite} information partitions: no matter how we extend~$\Omega$. This implies that the direct and indirect communication equilibria need no longer coincide almost surely, as was the case in~\cite{GePo1982}. If we consider even a more realistic situation without simplification~\ref{it:sim0}, 
then the result will in no way become more robust.



Thus, the result of the procedure under discussion may be almost arbitrary and uninterpretable in a setting where we are trying to aggregate extremely diverse information and arguments about a complex open problem. 
This leads us to the question whether there is a way to elicit experts' pooled information in an interpretable form, together with the probability of~$H$ given this information, in the case where the information structure is arbitrarily complex and is not commonly known.
In particular, for the above example the desired mechanism should always return the probability of~$H$ 
equal to $\pi(H\mid \neg A \cap \neg B)$, together with the information $\neg A \cap \neg B$ explicitly presented.

We emphasize once again that when discussing Bayesian truth serum, we are discussing the use of a very universal tool in a very specific situation for which it is not and should not be adapted. Therefore, our arguments in no way diminish the significance of this tool, and the specialized system we propose below cannot serve as a substitute for it.

\section{Preliminaries}
In prediction markets, participants can trade assets each of which binds a seller to pay a buyer $1$ unit of real or play money once a certain proposition is revealed to be true. Symmetrically, sellers keep the money received if the proposition is revealed to be false. For efficient matching of buyers and sellers, an automated market maker~\cite{Ha2007} can be applied (see also~\cite{OtSa2012}, where extremely flexible one is described). 
A well-known empirical fact is that clearing prices are well-calibrated in real-money prediction markets~\cite{PaCl2013} as well as in play-money ones~\cite{PeLaGi2001} in the sense of~\cite{DeGrFi1982}: prices approximate the real frequencies of true propositions. This fact by no means implies the convergence of prices in such markets to common posteriors or, even less, efficient information pooling, but it makes the corresponding questions reasonable.  
 In this connection, we note that the consensus theory \cite{Au1976,GePo1982} is naturally extended with the theory \cite{MKePa1986, NiBrGe1990, FeFoPe2005, ChMuCh2006} that explains why the price in such a market may converge to a common posterior for the proposition in question. Instead of the experts' beliefs being iteratively revealed and updated, the price dynamically aggregates them and, being commonly known, causes them to be dynamically updated. 
Additionally, in \cite{FeFoPe2005, ChMuCh2006} information pooling is studied theoretically in the context where information partitions are finite and commonly known.
 
 Since any aggregate of experts' beliefs cannot contain more information than the beliefs themselves, using a prediction market would not solve the problems described in the previous section. As above, information pooling need not even theoretically be efficient in the situation where the information structure cannot be simultaneously finite and commonly known. In addition, the result would not be interpretable. Moreover, new problems would arise for a conventional prediction market. First, to distribute experts' gains and losses after trading, we need to know the ground truth about the proposition in question. This fact severely limits the applicability of traditional prediction markets, especially for solving scientific problems. Second, there are experimental evidences \cite{MKePa1990,AlKiPf2009} that even in a situation, where efficient information pooling is theoretically possible, this may not be the case in reality. 
 One of the reasons may be the complexity of computations experts need to perform in order to elicit each other's information from the price.

 On the other hand, there are certain trivial situations, where (and only where) the efficiency of a prediction market is virtually guaranteed theoretically and experimentally. We mean the following basic property of prediction markets, which 
 can be considered as the simplest degenerate case of the theory~\cite{MKePa1986, NiBrGe1990}.
 \begin{enumerate}[label = (eff0)]
     \item\label{it:eff0}
     Suppose there are experts whose subjective probabilities of the truth of a certain proposition take the same value, and this fact is 
     a common belief among them. If we arrange, for these experts, a prediction market associated with the proposition in question, its final price will coincide with that value.
\end{enumerate}
This scenario is equivalent to the situation where experts share a common prior and all relevant information about the proposition is common knowledge. Under the assumptions of~\ref{it:eff0}, the persistent price coinciding with the experts' probabilities will be completely consistent with the experts' hierarchies of beliefs. We also note that if the common belief assumption is violated on any level of any hierarchy, then the hierarchies will eventually be updated and the probabilities (first-order beliefs) may change (together with the price).


In the mentioned study~\cite{AlKiPf2009}, property~\ref{it:eff0} has been confirmed experimentally. 
 We take a closer look at that study, where 
 the authors also supported the idea of utilizing play-money prediction markets with real incentives for solving scientific problems. More specifically, their experiments simulated the search for an answer
to the question ``In what order do three genes activate each other?'' Above, we have referred to the volatility of their experimental results in the setting where the information of each expert was private and transmitted
to others through the prices only. In such a setting, information pooling was broken in half the cases despite a finite (but rather complex) and commonly known information structure. On the other hand, in the setting where all the information was public, the final price perfectly approximated the conditional probability given that information in all the cases. The same was true for the third setting where the information was made public gradually by the experimenter. 

Is it possible to solve all the problems discussed above relying only on the simple and robust property~\ref{it:eff0} of prediction markets? Namely, can we arrange a system in which 
\begin{itemize}
\item 
private information would be directly made public by experts themselves (interpretability), 
\item 
we would be under conditions where only public information needs to be processed by the market (independence from the information structure and complex Bayesian calculations on it), and 
\item 
gains and losses would be distributed depending only on the information obtained, not on the ground truth (self-resolvability)? 
\end{itemize}
We demonstrate that such a system can be naturally obtained from a prediction market through minor modifications. 
We also provide a complete formalization of the solution through a theoretical model that is minimally necessary to capture the main ideas.

\section{Solution}\label{sec:solution}
Suppose we are interested in a certain scientific hypothesis~$H$ and would like to inspire its deep collective analysis by a large group~$E$ 
of experts. As the first (optional) step, we can publicly query an LLM for all available information relevant to~$H$ like ``Provide detailed and verifiable arguments for and against the scientific hypothesis [a description of~$H$]''. This step may get us closer to the situation of a common prior~$\pi$ over a certain sample space~$\Omega$ containing $H$ and of a substantial part of experts 
being ignorant, i.e., possessing only the public information provided by the LLM and incorporated in that prior. This also will lead us to a clearer distinction between public and private information. Namely, 
only the information that the experts perceive as private will remain private,
such as their unpublished results or original reasoning.

Next, we can arrange a play-money prediction market that will satisfy the following rules.
\begin{enumerate}[label = \textup{(r\arabic*)}]
    \item\label{it:r1} 
Each 
participant initially gets one and the same limited amount of play money. 
\item\label{it:r2} The trade inactivity of announced duration is a necessary condition for closing the market. 
\item\label{it:r3} 
Winners are determined by a binary random generator that implements Bernoulli distribution with probability parameter equal to the market final price. 
\item\label{it:r4} A public chat is provided for sharing any information directly.  
\end{enumerate}
These four rules are public and commonly known. 
In order to entangle such a market with the hypothesis~$H$, we can make the following public invitation to the experts. 
\begin{enumerate}[label = (inv0)]
\item\label{it:inv}
Trade as if the market would resolve in accordance with the truth of~$H$ and, 
after entering the market, share your private information on~$H$ through the chat in an interpretable and verifiable form. 
\end{enumerate}
It can be shown that such a course of action, publicly proposed to the experts, forms an equilibrium in the sense that if an expert believes 
the others follow invitation~\ref{it:inv} 
commonly believing in the same (i.e., an expert believes that any other expert follow~\ref{it:inv}, and believes that any other expert believes that any other expert follow~\ref{it:inv}, and so on), then he or she is bound, in order to make a profit, also to follow it. 
It is worth noting that such an equilibrium could be described in terms of separate strategies as a Bayesian Nash equilibrium. However, this would only complicate the model, at the same time making it less realistic. The fact is that for the required equilibrium to occur, experts do not need to consider each other's strategies. It is sufficient for each expert to reflect only on the probabilities that other experts assign to the generated binary outcome and on which they base their strategies.
The corresponding formal model will be provided in the final section, while here we discuss 
the main idea underlying it. 

Consider an arbitrary expert $n\in E$ with information~$I_n$ that is the element of his or her partition of $\Omega$ containing the true state of the world. For our basic scenario, 
we make certain simplifying assumptions 
that make information sharing structured and allow us to incorporate it into the formal model.
\begin{enumerate}[label = (shr\arabic*)]
\item\label{it:inf1} We assume that $n$ can publicly share $I_n$ with others in a completely interpretable and verifiable form, thereby making it common knowledge. 
\item\label{it:inf2} We also assume that the expert~$n$ cannot prove to others any false information, and that providing such unverifiable information does not affect others' beliefs and is identified with 
being silent. 
\item\label{it:inf3} Finally, we assume that the information~$I_n$ is indivisible in the sense that it cannot be transmitted and proved to others in parts in such a way that one part would increase the probability of~$H$ and the next part would decrease it, or vice versa. 
\end{enumerate}
Assumption~\ref{it:inf1} means that 
$n$ can reveal his or her partition of~$\Omega$ 
(e.g., describe an experiment he or she has conducted), report that its element~$I_n$ has turned out to be true (e.g., report the results of the experiment), and provide comprehensive evidence of this fact (e.g., provide the entire experimental data). Assumption~\ref{it:inf2} means that false information cannot be verifiable and unverifiable information has zero impact.
Concerning assumption~\ref{it:inf3}, we will see that it can be easily relaxed within the extended framework sketched in the next section, where we will also discuss 
the relaxation of~\ref{it:inf1}, \ref{it:inf2}, and other assumptions.

By $e_1\subseteq E$ we denote the subgroup 
of all ignorant experts, which are the experts who possess only the public information $I_n = \Omega$ (i.e., their knowledge is completely incorporated in the prior~$\pi$). 
We suppose it is common knowledge that such experts exist and that their number suffices to move the price as required anywhere below. We also note that due to~\ref{it:inf2}, ignorant experts will be 
silent: they have no private information that could be shared in verifiable form to become common knowledge.

Furthermore, we do not assume that experts know each other's information structures. Moreover, we assume that they do not speculate on them at all: 
each expert~$n$ considers $e_1\cup\{n\}$ 
to be the only participants in the market, at least 
as long as other experts 
do not enter it. 
In other words, each expert initially believes that there are many ignorant experts and does not go any further in his or her analysis. 
In particular, ignorant experts initially speculate only on the existence and beliefs of other ignorant experts. 
\begin{enumerate}[label = (blf1)]
\item\label{it:sim1}
Among ignorant experts, it is commonly believed that 
they are ignorant, and they do not have any beliefs about informed experts at all. 
\end{enumerate}

As we have noted above, even under experimental conditions where the information structure is common knowledge (but not the information itself), experts who are not specialists in Bayesian calculations often behave as if they do not know that structure at all~\cite{AlKiPf2009}. The unawareness assumption is all the more natural in complex settings where experts initially know nothing about each other. Nevertheless, in the next section we will discuss how this assumption 
can be dropped completely (including~\ref{it:sim1}). 
In addition, throughout this section it can be seen 
that all arguments and conclusions can be left unchanged even if 
we extend~\ref{it:sim1} as follows (cf. the Bayesian shift described in~\cite{Pr2004}: one's belief about someone else's belief on a subject shifts toward one's own belief on that subject).
\begin{enumerate}[label = (blf2)]
    \item\label{it:sim2} Each expert $n \in E$ believes that any other expert in~$E$ may either possess, with a certain small probability, 
    the same information~$I_n$, 
    or be ignorant. The expert~$n$ also believes that the others believe the same in the context of their own information, 
    and so on.  
\end{enumerate}

As long as only the experts from~$e_1$ trade, a desired equilibrium is naturally maintained. 
To see this, consider a certain expert~$n$ in~$e_1$. 
While~$n$ 
believes the others 
follow invitation~\ref{it:inv} commonly believing in the same, he or she 
expects them 
to jointly move the price to $\pi(H)$. 
Indeed, we have 
assumed that $n$ does not speculate 
on
the information of experts outside~$e_1$ 
before they enter the market (see~\ref{it:sim1}). 
Thus, $n$ believes the others will trade as in a prediction market 
where there is no private signals related to an outcome 
and where $\pi(H)$ is the commonly known probability of that outcome. 
Therefore, the basic property~\ref{it:eff0} will make $n$ 
to believe that the price will tend to $\pi(H)$ along 
with the probability parameter of the outcome generator. 
Furthermore, we can expect that~$n$ does not believe it is possible for him or her to distort the result of such concerted trading actions:
rule~\ref{it:r1} implies that several experts have more resources than a single one, rule~\ref{it:r2} implies they always have time to react to the price distortion, and, finally, 
$n$ hardly may expect 
to affect others' beliefs via a price shift. 
Concerning the last, suppose $n$ has shifted the price in a direction opposite to the concerted trading actions of the others. If the others believe that $n$ follows~\ref{it:inv}, then they will expect from him or her verifiable information that causes his or her probability to differ from $\pi(H)$. Since $n$ cannot provide such information, his or her price signal will disagree with~\ref{it:inv} and will be regarded as unverifiable manipulative information, which, by analogy with \ref{it:inf2}, will not affect the beliefs of the others. In general, it is hardly possible for a single expert with limited resources to manipulate the market via transactions, provided the number of participants is large (in the next section, we will discuss collusion and manipulation in more detail). 
We can conclude that  
rule~\ref{it:r3} will cause the expert~$n$ to believe that the probability of the outcome coincides with his or her probability of~$H$ and, as a consequence,  
to follow invitation~\ref{it:inv}.

As long as 
the experts from $e_1$ silently trade as described above, 
their beliefs derived from~\ref{it:sim1} as well as their common belief that everyone is following~\ref{it:inv}  
are 
supported by their observations of each other's behavior: a rational expectation equilibrium occurs.

Next, consider a structured scenario where a large number of ignorant experts 
stabilize the price at the level of~$\pi(H)$ before any actions of informed experts outside~$e_1$. 
Consider an expert $m\in E$ outside $e_1$ who has private information that affects the prior probability: 
$\pi(H\mid I_m)\ne \pi(H)$. 
As before, 
we assume that $m$ believes the others follow~\ref{it:inv} commonly believing in the same, 
and that he or she 
does not speculate on 
the existence of other informed experts  
before they enter the market. His or her belief that the current market participants are ignorant and follow~\ref{it:inv} is supported by their behavior. 
Thus, again, $m$ believes that
the current stable price~$\pi(H)$ is consensus for other experts, and that it is impossible, due to rules~\ref{it:r1} and~\ref{it:r2}, to distort it with his or her trading actions. Therefore, due to \ref{it:inf1}--\ref{it:inf3}, there exists the only course of action that $m$ believes to be profitable, provided he or she is risk neutral: to make a transaction at the current price~$\pi(H)$ and, after that, to share and prove the information~$I_m$, making it common knowledge. Indeed, by doing so, $m$ would recursively switch the system to the state 
with
\begin{itemize}
\item
$e_2 \df e_1\cup\{m\}$ instead of $e_1$,
\item 
$I_m$ instead of $\Omega$,
\item 
$\pi(\,\cdot\mid I_m)$ instead of $\pi$, and
\item 
$I_n \cap I_m$ instead of each $I_n$.
\end{itemize}
Assumption~\ref{it:sim1} will be also extended 
to this new state 
in the context of the updated public information~$I_m$ and the updated set~$e_2$ of ignorant experts: since it will be commonly believed that $m$ has followed~\ref{it:inv}, it will be commonly believed that $m$ has shared all of his or her information and has become ignorant. We also note that \ref{it:sim1} will encompass beliefs of~$m$ about ignorant experts. 

On the other hand, if $m$ enters the market without sharing information, the experts from $e_1$ will also extend their beliefs onto him or her: 
they will consider $m$ ignorant, which will be consistent with his or her silence and their belief that he or she is following~\ref{it:inv}. The expert~$m$ does not believe that under such conditions he or she will be able to change the price and make a profit (for the same reasons that give rise to the initial equilibrium among ignorant experts).

The above considerations imply 
that $m$ believes that the probability of the final outcome can be switched to  
the value $\pi(H\mid I_m)$, which differs from the current price~$\pi(H)$. Thus, $m$ believes that he or she can make a profit and that there is the only way to do it. 

It is not hard to see that even if $m$ admits that 
another expert may, with small probability, possess the same information~$I_m$ (see~\ref{it:sim2}), it is still profitable for $m$ to act as described above. In particular, $m$ 
believes that after he or she shares the information, all experts will become ignorant in the context of the updated public information and will stabilize the price at the level of~$\pi(H\mid I_m)$.

Everything that has been said about $m$ is true for any other expert whose private information is inconsistent with the prior probability. 
Regardless of which one of them manages to be the first to leverage
his or her informational advantage,
the system will be switched to a new state, similar to the initial one, but 
with updated public information. 

Above, we have also assumed the following (cf. the infinite population assumption in~\cite{Pr2004}). 
\begin{enumerate}[label = (eff1)]
\item\label{it:inst}
A number of ignorant experts is large, and in the equilibrium described above, they quickly balance the price in accordance with 
public information: 
they manage to do this before any informed expert enters the market. 
\end{enumerate}
This will be true, for example, if informed experts enter the market at intervals, or 
the market satisfy the semi-strong efficient market hypothesis that the price instantly reflects all the public information (ignorant experts instantly implement~\ref{it:eff0}). The assumption~\ref{it:inst} allows us to make the formal model parsimonious, but it can be seen that its violation will not affect the system results. We discuss the relaxation of this and other assumptions in the next section.

Reasoning recursively under assumptions made throughout this section, 
we arrive at the following result, whose complete formalization 
will be presented in the final section.
\begin{enumerate}[label = (T1)]
    \item\label{it:T1} Consider the system with rules \ref{it:r1}--\ref{it:r4}. 
    If it is commonly believed that the others follow~\ref{it:inv}, 
    then all the experts will indeed follow~\ref{it:inv}. The absence of trading activity will mean that there remain no experts whose private information 
    make his or her probability of~$H$ 
    diverged from the current price.
\end{enumerate}

Now we discuss information pooling. Again, we are going to simplify the situation in some way, but we will discuss the general case in the next section.

We say that information $I_n$ is \emph{a direct argument} \emph{for} or \emph{against} the hypothesis~$H$ if $I_n$ forms a pair of nested sets with~$H$ or~$\neg H$, respectively. For example, $I_n = \Omega$ is a degenerate argument, while 
information $I_n \subseteq H$ or $I_n \subseteq \neg H$ is a conclusive proof of $H$ or $\neg H$, respectively. 
In other cases, where either $I_n \supset H$ or $I_n \supset \neg H$, a direct argument 
rules out some variants of realization of either~$\neg H$ or~$H$, respectively.  

If all private information consists of direct arguments, this will remain true after each information-sharing episode. It is not hard to see, that this fact yields the following result, whose rigorous formulation and proof can be found in the final section.
\begin{enumerate}[label = (T2)]
\item\label{it:T2} If each $I_n$ is a direct argument {for} or {against} the hypothesis~$H$, then under the conditions of~\ref{it:T1}, at least one of the following two conclusions is true: 
\begin{itemize}
\item
one of the hypotheses $H$ or $\neg H$ will be proven publicly and conclusively, 
\item
all arguments~$I_n$ will become public. 
\end{itemize}
\end{enumerate}


Now we recall our introductory example~\ref{it:ex} and consider applying our system to it, under the assumptions made throughout this section. 
We immediately get into the conditions of~\ref{it:T2}. Thus, both arguments $\neg A$ and $\neg B$ will become public and the final price will coincide with $\pi(H\mid \neg A \cap \neg B)$.

In the next section, we will also discuss why the condition of~\ref{it:T2} is redundant (in addition to the lack of necessity for \ref{it:sim1} or \ref{it:sim2}, as well as for other assumptions).

\section{Robustness}
Here we provide strong, but not completely rigorous arguments 
that 
the violation of 
the assumptions we have made in the previous section will not affect the system results. 
A complete formalization of the considerations below could be the subject of future research.

We start with assumption~\ref{it:inst}. 
Suppose it takes some time for the price to stabilize.
For any informed expert who is 
observing the market, 
it is profitable to start trading as ignorant one 
and, after the price stabilization, 
to leverage his or her informational advantage as described in the previous section. This is true regardless of the relative order of the following three quantities: 
\begin{enumerate}[label = (q\arabic*)]
\item\label{it:q3} the starting price, 
\item\label{it:q2} the probability of $H$ conditioned on the current public information, and
\item\label{it:q1} the informed expert's probability of $H$. 
\end{enumerate}
Indeed, the informed expert believes that if \ref{it:q2} does not lie between \ref{it:q3} and~\ref{it:q1}, then it is possible to increase the length of the price trajectory by allowing the price to first reach \ref{it:q2} and then sending it to \ref{it:q1}, by sharing his or her information. By making transactions in accordance with such price movements, the expert will be able to maximize profits. 
In any case, the price passes, as before, through quantity~\ref{it:q2}, quantity~\ref{it:q1} iteratively becomes quantity~\ref{it:q2} after the 
information-sharing episode, and so forth: we arrive at the same results.
Moreover, such behavior of experts is an equilibrium: if it is commonly believed that during a round between information-sharing episodes, all experts behave as ignorant ones until the price stabilizes, then it is disadvantageous to deviate from such behavior. 

In addition, we can drop our assumptions on experts' beliefs about each other's information structures over~$\Omega$ 
(including assumption~\ref{it:sim1}). 
We only need to assume that it is commonly believed that the number of ignorant experts suffices to ensure that their concerted trading actions in the sense of~\ref{it:eff0} cannot be distorted by others' trading actions (such actions can be assumed not to be personified). 
This guarantees that no matter how experts speculate about others' information structures, they will believe that after the price stabilizes at the probability conditioned on the current public information, it can only be shifted through sharing their private information, which makes such actions profitable. We only need to justify that the equilibrium among ignorant experts will be maintained. But it is not hard to see that no matter how an ignorant expert speculates on others' information structures, his or her probability of $H$ varying due to information exchange will be believed to be a martingale, regardless of his or her beliefs about the order in which informed experts will share information. 
This implies that if such an expert believes that in each round between information sharing episodes, the other ignorant experts are under the conditions of~\ref{it:eff0} with respect to the probability of $H$ conditioned on the current public information, then he or she will also believe that the price at stability points will be a martingale starting at $\pi(H)$. This means that his or her probability of the generated outcome will always coincide with the probability of $H$, conditioned on the current public information. Thus, the required equilibrium occurs. 

We note that all our formal assumptions about hierarchies of beliefs are aimed at modeling a natural situation where an expert with significant information assumes that due to limited resources and the expected privacy of this information, he or she, as well as someone else with similar information, will not be able to incorporate it into the price solely through trading. Below, we will extend this situation with information disseminated among experts. 

It is also worth noting that in the extended settings under discussion, 
we can 
relax~\ref{it:inf3}, provided we will assume that any expert possesses a finite number of indivisible units of information. Such an expert can use these units one after another, maximizing the length of the price trajectory, but, as before, he or she will stop sharing information only when his or her probability of $H$ 
becomes consistent with the price at stability points. 

In addition, assuming that only some of the information units are transferable and verifiable, we get a basis for relaxing~\ref{it:inf1}. We will return to this point below.


Next, we can bring invitation~\ref{it:inv} into line with the new extended equilibrium. 
\begin{enumerate}[label = (inv1)]
\item\label{it:inv1}
Trade as if the probability of the outcome would coincide with the probability of~$H$ conditioned on the current public information. 
If the price becomes fair with regard to this probability, you may 
disclose any part of your information simultaneously with making a transaction that would be 
determined by the probability of~$H$ conditioned on the updated public information. 
\end{enumerate}
The behavior described in~\ref{it:inv1} is indistinguishable from the one described in~\ref{it:inv} only under assumptions~\ref{it:inf3} and~\ref{it:inst}. 
In these sense, \ref{it:inv1} extends~\ref{it:inv}. 
However, even if we leave the more straightforward invitation~\ref{it:inv} unchanged and there occur 
experts with private information who follow it literally, we can expect that the system will not cease to be effective. Indeed, if such an expert enters the market and shares his or her information without waiting for the market to process the current public information, he or she will only accelerate the recursively switching of the system to a new state with updated public information. 

Now we discuss how it may be possible to relax assumption~\ref{it:inf1} together with the assumption that there are many experts with only public information. Suppose $\pi$ is a common prior only over an extended space $\Omega \times \Gamma$. Here, information of each expert over $\Gamma$ is his or her background knowledge that is disseminated among experts and difficult to distinguish from 
public information. In particular, if elements of $\Gamma$ are represented as arrays of bits, then we can assume that each such bit is known to many experts at once. Knowledge over $\Gamma$ could be responsible for the correct interpretation of information over~$\Omega$. For example, this could be knowledge of how to correctly perform statistical and Bayesian calculations. We assume that information over $\Gamma$ is not directly transferable (including because experts cannot distinguish it from public information). By ignorant experts, we continue to mean experts who possess information of the form $I_n = \Omega \times \gamma$, where $\gamma \subseteq\Gamma$. Thus, ``ignorant'' does not mean ``uneducated'' in our terminology. 
It can be expected that, in the equilibrium, experts will, first, naturally 
rely on their background knowledge over $\Gamma$, and second, try to capture prevalent homogeneous price signals generated by the background knowledge of other experts: such behavior will be profitable provided it is commonly believed that others are doing the same. 
Thus, information that experts subjectively consider to be private 
and non-distributable through price signals will be distributed directly through the chat, while the rest of the information will be distributed through the market mechanism. The fact that such a mechanism works for background knowledge complements property~\ref{it:eff0} and has been empirically verified in the aforementioned article~\cite{AlKiPf2009} together with~\ref{it:eff0} itself: the experts who participated in that experiment could have different levels of proficiency in the mathematical apparatus necessary to interpret the simulated results of gene activation tests. 

Concerning \ref{it:T2}, we can impose more subtle conditions on the relationship between experts' information and the hypothesis~$H$ that will guarantee information pooling regardless of the probability measure. For example, we can expect that if the outcomes in $\Omega$ are represented as arrays of bits, and each expert knows a finite number of bits of the real state of the world, then such a condition could be the expressibility of the hypothesis $H$ through a threshold Boolean function on $\Omega$ as in~\cite{FeFoPe2005, ChMuCh2006}. More generally, 
the following situation seems rare: the remaining private information can change the probability of~$H$, but it is divided among the experts in such a way that each part 
does not change this probability. We can expect that the rarity of such a situation can be formalized similarly to how the rarity of a discrepancy between the direct and indirect communication equilibria was formalized in~\cite{GePo1982}.


As for the possibility of manipulations, 
consider an extreme situation where all players are coordinated by a single center. In the worst case, this is equivalent to a situation where a single player with combined resources plays against an automatic market maker in a conventional prediction market, with guaranteed completion of trading at a fair price, matching the frequentist probability of the event in question. The more experts participate in the system, the less likely such a situation is to occur in practice, and its consequences will depend on the properties of the market maker, whose potential losses are limited in any case~\cite{OtSa2012}. 
In general, we can expect that our system will be no less robust to manipulative trading attacks than conventional prediction markets~\cite{HaOpPo2006,JiSa2012} are. Moreover, due to interpretability, manipulative transactions without supporting information will have little effect, and after the trade is completed, it will be possible to verify the consistency between the final price and the final public information. Concerning manipulations through false information, we note that even when \ref{it:inf2} is violated and such a manipulation affects the price, it becomes profitable for everyone (including the manipulator) to reliably negate the manipulative information if possible. This allows us to expect that \ref{it:inf2} can also be relaxed. 

\begin{description}
    \item[Conclusion] 
    It can be expected that property~\ref{it:eff0} is all we need to entangle our system with any scientific hypothesis and to achieve information pooling.   
\end{description}

\section{Formalization} 
\label{sec:model}
In this section, we provide a complete formalization of our minimal model.

Consider a probability space
$
(\Omega, \mathcal{F}, \pi),
$
where $\Omega$ represents possible states of the world and $\pi$ represents the common prior of the experts from~$E$.
Let $\mathcal{I}_n\subseteq \mathcal{F}$ be partitions of $\Omega$ that represent the private information of the corresponding experts $n\in E$. Suppose $\omega_0\in\Omega$ is the true state of the world.
We do \emph{not} assume that experts from $E$ are initially aware of each other's information partitions or even of each other's existence. 

In what follows, only the elements $I_n \in \mathcal{I}_n$ such that $I_n \ni w_0$ will be involved in the information exchange between experts, and we will refer to them as the private information of experts 
(instead of~$\mathcal{I}_n$).
We also make the standard technical assumption that $\pi(\cap_{n\in E}\, I_n) \ne 0$. 
Next, we introduce the set
$$
e_1\df \{n \in E \mid I_n = \Omega\} 
$$
of the experts who initially possess only the public information, which is beyond~$\Omega$ and is not explicitly presented. 


In our theoretical model, we regularize wagering and information exchange between the experts as follows. 
We assume that experts in $e_1$ 
generate an initial prediction market. We consider their trading as the first round that results in a price $\xi_{1}$. 
After the first round, some new expert $n \in E\setminus e_1$ may enter the market by making a transaction at the price~$\xi_1$. After that this expert~$n$ may (or may not) share his or her private information 
as it is described 
just below. 
Next, the second round begins, and the experts in 
$e_{2} = e_1\cup\{n\}$ are trading with each other. It results in a price $\xi_{2}$ and so forth: the process continues until there remain no experts willing to enter the market. Let $k_\infty$ be the number of the final round, which obviously exists because~$E$ is finite.
Thus, we finally have the sequences
$$
e_1 \subset \dots \subset e_{k_\infty} \subseteq E,\quad |e_{k+1}\setminus e_k|=1,
$$
and
$$\xi_1,\dots,\xi_{k_\infty} \in \Xi \df [0,1].$$
After the final round~$k_\infty$, we produce 
a value
$$
    \theta \in \Theta \df \{0,1\}
$$
by means of a binary random generator with the probability parameter~$\xi_{k_\infty}$. We resolve the market \textup{(}distribute gains and losses\textup{)} in accordance with 
    $\theta$: each contract brings~$1$~unit if $\theta = 1$ and nothing if $\theta = 0$. 
    We introduce the sample space 
    $$\Lambda \df \Xi\times \Theta,$$
where each element $(\xi,\theta) \in \Lambda$ represents a 
\emph{final} price together with a \emph{final} result produced from this price by a random generator.

Next, we introduce a notation for information exchange between experts (under assumptions \ref{it:inf1}--\ref{it:inf3}). 
The initial public data is $\Omega_1 \df \Omega$. 
As described above, if 
after a 
round~$k$ 
a new expert $n\in E\setminus e_k$ enter the market by making a transaction at the price $\xi_k$, then after that he or she can choose whether to remain silent, or to share his or her information.
The choice of silence leads to
\begin{equation}\label{eq:silence}
        \Omega_{k+1} \df \Omega_k,
\end{equation}
whereas the choice of sharing leads to
    \begin{equation}\label{eq:sharing}
        \Omega_{k+1} \df I_n^k, 
    \end{equation}
where 
$$
I^k_n \df \Omega_k \cap I_n.
$$

Thus, we have two separate mechanisms: wagering 
with its final outcome from~$\Lambda$, and information exchange over 
$\Omega$. Our goal is to show how such mechanisms can be entangled 
with each other in order to get
the probability of a hypothesis
$H\subseteq \Omega$ and to collect the relevant information interpreting that probability. 

Our subsequent formal assumptions rely inter alia on the idea that during a round~$k$ an expert $n\in E$ 
may know nothing
about the experts outside $e_k\cup\{n\}$, i.e.,
about their information 
(and its structure), about their intentions to enter the market, or even about their existence: we assume that 
$n$ considers only ignorant experts, the current participants, and himself or herself as real players. In particular, the experts $n\in e_k$ consider the round~$k$ as final.

The following 
assumption~\ref{Asmp:mu} means that {at the beginning} of a round~$k$, the experts in $e_k$ have hierarchies of beliefs over $\Theta$ as well as first-order beliefs over~$\Lambda$, agreed with these hierarchies. 
In accordance with the idea that experts in~$e_k$ do not know much about undiscovered individuals outside~$e_k$, their hierarchies of beliefs apply only to~$e_k$ itself.
We introduce them through Harsanyi's types \cite{Ha1967,MeZa1985}. 
Following~\cite{MeZa1985}, we denote by $\Pi(X)$ the compact space of probability
measures on a compact space $X$. 
\begin{Asmp}\label{Asmp:mu} For~$\Theta$ and the players in~$e_k$, we introduce 
the universal type space $T_k$ with the corresponding homeomorphism 
$$h_k\colon T_k \to \Pi\bigl(\Theta \times T_k^{|e_k|-1}\bigr).$$
We assume that for each $n \in e_k$, we can assign a type $t^k_n \in T_k$ and a belief $\mu_n^k \in \Pi(\Lambda)$
such that
\begin{equation}\label{eq:marg}
\marg_\Theta \mu_n^k = \marg_\Theta h_k\big(t^k_n\big)
\end{equation}

\end{Asmp} 
The next 
assumption 
means that experts know~\ref{it:r3}. 
\begin{Asmp}\label{Asmp:coh}
    Suppose $n\in e_k$ and $\rho \in [0,1]$.
    If 
    $\mu_n^k(\xi = \rho) \ne 0,$ then
    $$\mu_n^k(\theta = 1 \mid \xi = \rho) = \rho.$$
\end{Asmp}
Hereinafter, we often use the same symbol to refer to a measure and to its marginals. In particular, by $\mu_n^k(\xi = \rho)$ and $\mu_n^k(\theta = 1)$ we mean 
$$\marg_{\Xi}\mu_n^k\bigl(\{\rho\}\bigr) = \mu_n^k\bigl(\{\rho\}\times\Theta\bigr) 
\quad\mbox{and}
\quad \marg_{\Theta}\mu_n^k\bigl(\{1\}\bigr) = \mu_n^k\bigl(\Xi\times\{1\}\bigr),
$$ 
respectively.

\begin{Def}\label{Def:type_rho}
    Suppose $n\in e_k$ and $\rho \in [0,1]$. We define the type $\langle\rho\rangle_k \in T_k$ as follows. We say that $t_n^k = \langle\rho\rangle_k$ if
    $$
    h_k\bigl(t_n^k\bigr)(\theta = 1) = \rho
    $$
    and
    \begin{equation}\label{eq:conc}
        h_k\bigl(t_n^k\bigr)\bigl(t_m^k = \langle\rho\rangle_k \;\;\forall\; m \in e_k\setminus\{n\}\bigr) = 1.
    \end{equation}
\end{Def}
\begin{Def}\label{Def:type_rho_comm}
    If the type $t_n^k$ of an expert $n \in e_k$ satisfies~\eqref{eq:conc} for some $\rho \in [0,1]$, then we say that 
    $n$ believes 
    the others in $e_k$ have the type~$\langle\rho\rangle_k$.
\end{Def}
The next model 
assumption is the basic property~\ref{it:eff0}. 
\begin{Asmp}\label{Asmp:same_prob}
    Let $\rho \in [0,1]$. If $t_n^k = \langle\rho\rangle_k$ for all experts $n\in e_k$, then $\xi_k = \rho$.
\end{Asmp}


Next, we formalize the assumption we justified in Section~\ref{sec:solution} 
based on~\ref{it:r1}, \ref{it:r2}, and~\ref{it:inf2} (or 
merely on the assumption that the number of market participants is large): if at the beginning of round~$k$, an expert $m\in e_k$ believes that all the other current participants satisfy the premise of Assumption~\ref{Asmp:same_prob} for some $\rho$, then $m$ believes that he or she will not be able to resist others' concerted actions and the round~$k$, which he or she believes to be final, will result in the price~$\rho$.
\begin{Asmp}\label{Asmp:rbst}
    Suppose $n\in e_k$ and $\rho \in [0,1]$. If $n$ believes the others in $e_k$ have 
    the type~$\langle\rho\rangle_k$, 
    then $\mu_n^k(\xi = \rho) = 1$.    
\end{Asmp}

The final assumption~\ref{Asmp:joining} consists of two theses. First, a new expert enters the market only if he or she believes that the corresponding transaction yields a positive expected profit, i.e., that the current price is unfair. Second, an expert who believes that the current price is unfair, is willing to enter the market. In order to formalize the second thesis, 
we 
implicitly assume than $n$'s beliefs on $\Lambda$ that guide 
his or her decision to join~$e_k$, 
will be identified with $\mu_n^{k+1}$ in a situation where 
he or she actually joins~$e_k$ and $e_{k+1} = e_k\cup\{n\} $. 
\begin{Asmp}\label{Asmp:joining}
\leavevmode 
    \begin{enumerate}
        \item\label{it:joining1} If $e_{k+1} = e_k\cup\{n\}$, then 
        \begin{equation}\label{eq:can_gain}
        \mu_n^{k+1}(\theta = 1) \ne \xi_k.
        \end{equation}
        \item\label{it:joining2} If there exists an expert $n \in E\setminus e_k$, 
        for whom joining~$e_k$ and choosing one of alternatives~\eqref{eq:silence} or~\eqref{eq:sharing} 
        would necessarily lead to \eqref{eq:can_gain}, then 
        $k\ne k_\infty$. 
    \end{enumerate}
\end{Asmp}
We note that the model does not specify the order in which experts enter
the market, but only assumes that 
trading cannot stop if there is an expert who
believes that he or she can make a profitable transaction.


Suppose a state sequence
$$
\mathcal{M} \df \bigl\{(e_k,\xi_k,\Omega_k)\bigr\}_{k=1}^{k_\infty}
$$
satisfies Assumptions~\ref{Asmp:mu}--\ref{Asmp:joining}. Suppose   
$H\subseteq \Omega$. 
We can formally describe a situation, where everyone follows~\ref{it:inv}, this is common knowledge, and, as a consequence, \ref{it:sim1} is maintained throughout all rounds.
\begin{Def}
We say that $\mathcal{M}$ 
is 
entangled with~$H$ if the following are true for any $k = 1,\dots,k_\infty$. 
\begin{enumerate}[label = \textup{(ent\arabic*)}]
    \item\label{it:eff1} If $n \in e_k$, then 
    $I_n^k = \Omega_k$.
    \item\label{it:eff2} If $n \in e_k$, then 
    $t_n^k = \langle\pi(H \mid \Omega_k)\rangle_k.$
    \item\label{it:eff3}
    \begin{enumerate}
    \item\label{it:eff3_1} If $e_{k+1} = e_k\cup\{n\}$, then 
        \begin{equation}\label{eq:can_gain_pi}
        \pi\bigl(H \bigm| I_n^k\bigr) \ne \xi_k.
        \end{equation}
    \item\label{it:eff3_2} If $n \in E\setminus e_k$ and \eqref{eq:can_gain_pi} holds, then $k \ne k_\infty$.
    \end{enumerate}
\end{enumerate}
\end{Def}
It is easily seen that the final state in such $\mathcal{M}$ is described as follows.
\begin{Fact}\label{Fact:final_state}
    If $\mathcal{M}$ 
    is entangled with 
    $H$, then
    \begin{equation}\label{eq:final_state_inf}
         \Omega_{k_\infty} = \bigcap_{n\in e_{k_\infty}} I_n
    \end{equation}
    and
    \begin{equation}\label{eq:final_state_prob}
    \xi_{k_\infty} = \pi\bigl(H \bigm| \Omega_{k_\infty}\bigr) = \pi\bigl(H \bigm| I_n^{k_\infty}\bigr),\quad n \in E.
    \end{equation}
\end{Fact}
\begin{proof}
Due to~\ref{it:eff1}, we have
$$
\Omega_{k_\infty} = \bigcap_{n \in e_{k_\infty}} I_n^{k_\infty} = \Omega_{k_\infty}\cap\bigcap_{n \in e_{k_\infty}} I_n.
$$
Thus, we have
    $$
         \Omega_{k_\infty} \subseteq \bigcap_{n\in e_{k_\infty}} I_n.
    $$
The reverse inclusion is always true by the construction of $\Omega_{k_\infty}$.

The first identity in~\eqref{eq:final_state_prob} follows from~\ref{it:eff2} and Assumption~\ref{Asmp:same_prob}. The second identity in~\eqref{eq:final_state_prob} 
follows from~\ref{it:eff1} for $n\in e_{k_\infty}$ and from item~\ref{it:eff3_2} of~\ref{it:eff3} for $n \in E\setminus e_{k_\infty}$. 
\end{proof}

Now we are ready to formulate and prove the theorem that, together with Fact~\ref{Fact:final_state}, is a formalization of~\ref{it:T1}. Its condition combines belief~\ref{it:sim1} with the common belief that the
others follow~\ref{it:inv} (in particular, sharing their information when entering the market).
Its conclusion means that all beliefs turn out to be true: a rational expectation equilibrium occurs.
\begin{Th}[formalization of~\ref{it:T1}]
Suppose that whenever $n\in e_k$, $n$ believes the others in $e_k$ have 
the type $\langle\pi(H \mid \Omega_k)\rangle_k$. 
Then $\mathcal{M}$ 
is 
entangled with~$H$.
\end{Th}
\begin{proof}
First, combining Assumptions~\ref{Asmp:rbst} and~\ref{Asmp:coh}, we conclude that
if $n \in e_k$, then
\begin{equation}\label{eq:main}
\mu_n^k(\theta = 1) = \pi(H \mid \Omega_k).
\end{equation}
In particular, by~\eqref{eq:marg} and by Definitions~\ref{Def:type_rho} and~\ref{Def:type_rho_comm}, 
we have $t_n^k = \langle\pi(H \mid \Omega_k)\rangle_k.$ We have proved~\ref{it:eff2}.

    Now we prove~\ref{it:eff1}.
    If $n \in e_1$, then 
    $$
        I_n^1 = \Omega_1 \cap \Omega = \Omega_1.
    $$    
    Next, suppose $n$ enters the market between rounds $k-1$ and $k$: 
    $$e_k = e_{k-1} \cup \{n\}.$$ 
    Then, by item~\ref{it:joining1} of Assumption~\ref{Asmp:joining}
     we have $$\mu_n^k(\theta = 1) \ne \xi_{k-1}.$$
    Therefore, by~\eqref{eq:main}, we have 
    \begin{equation}\label{eq:proof1}
    \pi(H \mid \Omega_k)\ne \xi_{k-1}.
    \end{equation}
    Further, by \ref{it:eff2} for $n \in e_{k-1}$ and by Assumption~\ref{Asmp:same_prob}, we have 
    $$\xi_{k-1} = \pi\bigl(H \bigm| \Omega_{k-1}\bigr).$$
    This, together with~\eqref{eq:proof1}, implies that $\Omega_k \ne \Omega_{k-1}$. 
    Thus, the only alternative remaining is~\eqref{eq:sharing}: 
    \begin{equation}\label{eq:proof2}
        \Omega_k = I^{k-1}_n \df \Omega_{k-1} \cap I_n.
    \end{equation} 
    Therefore, 
    $\Omega_k \subseteq I_n$ and we have $I^k_n = \Omega_k$. Since $\Omega_l \subseteq \Omega_k \subseteq I_n$ for $l\ge k$, we have 
    proved~\ref{it:eff1}. In addition, combining~\eqref{eq:proof1} and~\eqref{eq:proof2}, 
    we come to item~\ref{it:eff3_1} of~\ref{it:eff3}. 

    It remains to prove item~\ref{it:eff3_2} of~\ref{it:eff3}.
    Suppose an expert $n \in E \setminus e_k$ satisfies~\eqref{eq:can_gain_pi}. If~$n$ joins~$e_k$ and
    chooses alternative~\eqref{eq:sharing}, then due to~\eqref{eq:main} we would necessarily have
    $$
        \mu_n^{k+1}(\theta = 1) = \pi(H \mid \Omega_{k+1}) = \pi\bigl(H \bigm| I_n^k\bigr) \ne \xi_k.
    $$
    Thus, by item~\ref{it:joining2} of Assumption~\ref{Asmp:joining} we have $k\ne k_{\infty}$. This completes the proof.
\end{proof}

Finally, we discuss information pooling in a situation where experts' information consists of direct arguments. 
\begin{Th}[formalization of~\ref{it:T2}]
Suppose $\mathcal{M}$ is entangled with $H$.
\begin{equation*}
\mbox{If for each $n \in E$, we have}\quad
\left[
    \begin{array}{l} 
I_n \subseteq  H, \\ 
I_n \supseteq  H, \\
I_n \subseteq  \neg H, \\
I_n \supseteq  \neg H,
\end{array}
\right.\quad\mbox{then}\quad
\left[
    \begin{array}{l} 
\Omega_{k_\infty} \subseteq  H, \\ 
\Omega_{k_\infty} \subseteq  \neg H, \\
\Omega_{k_\infty} = \bigcap_{n\in E} I_n,
\end{array}
\right.
\end{equation*}
where all the relationships 
are fulfilled up to a null set.
\end{Th}
\begin{proof}
    Suppose the conclusion of the theorem is not true. Thus, we have the following:
    $$
        \pi(H \cap \Omega_{k_\infty}) \ne 0,\qquad \pi(\neg H \cap\Omega_{k_\infty})\ne 0,
    $$
    and it is \emph{not} true that 
    $$
        \Omega_{k_\infty} \subseteq \bigcap_{n\in E} I_n
    $$
    up to a null set (because the reverse inclusion is always true). The last means that there exists $n\in E$ such that it is \emph{not} true that $\Omega_{k_\infty} \subseteq I_n$ up to a null set. 
    Putting it all together, we get 
    \begin{equation}\label{eq:by_contr}
   \begin{cases}
\pi(H \mid \Omega_{k_\infty}) \ne 0, \\ 
\pi(H \mid \Omega_{k_\infty}) \ne 1, \\
\exists\,n \in E 
\colon\; \pi(\Omega_{k_\infty} \cap I_n) \ne \pi(\Omega_{k_\infty}).
\end{cases} 
\end{equation}

For such an expert~$n$, consider the cases where $I_n \subseteq  H$ or $I_n \subseteq  \neg H$ up to a null set.
In these cases, $\pi\bigl(H \bigm| I_n^{k_\infty}\bigr)$ equals either $1$ or $0$, respectively. This fact, together with the first two relationships in~\eqref{eq:by_contr}, contradicts the second identity in~\eqref{eq:final_state_prob}.

Next, suppose $I_n \supseteq H$. Due to the third relationship in~\eqref{eq:by_contr}, we get 
\begin{equation}\label{eq:cond_prob}
\pi\bigl(H \bigm| I_n^{k_\infty}\bigr) = \frac{\pi(\Omega_{k_\infty} \cap H)}{\pi(\Omega_{k_\infty} \cap I_n)} 
\ne \frac{\pi(\Omega_{k_\infty} \cap H)}{\pi(\Omega_{k_\infty})} = \pi(H \mid \Omega_{k_\infty}).
\end{equation}
Again, we arrive at a contradiction with the second identity in~\eqref{eq:final_state_prob}.

Finally, if $I_n \supseteq \neg H$, we can substitute $H$ for $\neg H$ in~\eqref{eq:cond_prob}.
We get 
$$
    \pi\bigl(\neg H \bigm| I_n^{k_\infty}\bigr) \ne \pi(\neg H \mid \Omega_{k_\infty}).
$$
But this immediately implies the same relationship for $H$ instead of $\neg H$.
Thus, we arrive at the same contradiction as before and complete the proof.
\end{proof}

\section{Acknowledgments}
We are grateful to Alexander Nesterov, Dmitry Rutsky, Fedor Sandomirskiy, and Dmitry Spelnikov for imparting knowledge, fruitful discussions, and valuable comments.

\printbibliography

\end{document}